\documentclass[10pt,conference,letterpaper]{IEEEtran}

\IEEEoverridecommandlockouts

\usepackage[table,usenames,dvipsnames]{xcolor}      % color
\usepackage{amsmath,amssymb,amsfonts,amsthm,dsfont} % math
\usepackage{algorithm,algorithmicx,listings}        % algorithms  
\usepackage{algcompatible}
\usepackage{enumitem}

\usepackage{lineno}
\usepackage{times,epsfig}
\usepackage{subfigure}
\usepackage{empheq}
\usepackage{graphicx}
\usepackage{fge}
\usepackage{mathrsfs}
\usepackage{setspace}
\usepackage{dblfloatfix}
\usepackage{tikz}
\usepackage{tkz-tab}
\usepackage{lipsum}
\usepackage{mathrsfs}
\usepackage{epstopdf}
\usepackage{cite}
\usepackage[breaklinks=true, bookmarks=true, colorlinks,
            citecolor=Black, urlcolor=Black]{hyperref}

\floatname{algorithm}{Algorithm}
\newtheorem{mydef}{Definition}
\newtheorem{myassump}{Assumption}
\newtheorem{mytheorem}{Theorem}
\newtheorem{mylemma}{Lemma}
\newtheorem{myremark}{Remark}
\newtheorem{mycorollary}{Corollary}
\newtheorem{myproposition}{Proposition}

\newtheorem{myproblem}{Problem}

\newtheorem*{mysystem}{System Model}
\newtheorem*{mymeas}{Measurement Model}
\newtheorem*{mysensorsch}{Sensor Scheduling Model}

\modulolinenumbers[5]

\newcounter{ale}

\newenvironment{liste}{\begin{itemize}}{\end{itemize}}
\newcommand{\aliste}{\begin{liste} \setcounter{ale}{1}}
\newcommand{\zliste}{\end{liste}}

\newcommand{\myparagraph}[1]{\vspace{0.5ex}\noindent\textbf{#1:}}

\title{{\LARGE {\bf Scheduling Nonlinear Sensors for Stochastic Process Estimation}}}
\author{Vasileios Tzoumas{$^{\star}$}, Nikolay A.~Atanasov{$^{\star}$}, Ali Jadbabaie{$^{\dagger}$}, George J.~Pappas{$^{\star}$}
\thanks{$^{\star}$The authors are with the Department of Electrical and Systems Engineering, University of Pennsylvania, Philadelphia, PA 19104-6228 USA (email: {\fontsize{8}{8}\selectfont\ttfamily\upshape \{vtzoumas, atanasov, pappasg\}@seas.upenn.edu}).}
\thanks{$^{\dagger}$The author is with the Department of Civil and Environmental Engineering, Massachusetts Institute of Technology, Cambridge, MA, 02139 USA (email: {\fontsize{8}{8}\selectfont\ttfamily\upshape jadbabai@mit.edu}).}
\thanks{This work was supported in part by TerraSwarm, one of six centers of STARnet, a Semiconductor Research Corporation program sponsored by MARCO and DARPA, in part by AFOSR Complex Networks Program and in part by AFOSR MURI CHASE.}
}

\begin{document}
\maketitle

% This paper focuses on choosing subsets of sensors to use over time in order to estimate the state of a stochastic process of interest.

\begin{abstract}
In this paper, we focus on activating only a few sensors, among many available, to estimate the state of a stochastic process of interest. This problem is important in applications such as target tracking and simultaneous localization and mapping (SLAM). It is challenging since it involves stochastic systems whose evolution is largely unknown, sensors with nonlinear measurements, and limited operational resources that constrain the number of active sensors at each measurement step. We provide an algorithm applicable to general stochastic processes and nonlinear measurements whose time complexity is linear in the planning horizon and whose performance is a multiplicative factor $1/2$ away from the optimal performance. This is notable because the algorithm offers a significant computational advantage over the polynomial-time algorithm that achieves the best approximation factor $1/e$. In addition, for important classes of Gaussian processes and nonlinear measurements corrupted with Gaussian noise, our algorithm enjoys the same time complexity as even the state-of-the-art algorithms for linear systems and measurements. We achieve our results by proving two properties for the entropy of the batch state vector conditioned on the measurements: a) it is supermodular in the choice of the sensors; b) it has a sparsity pattern (involves block tri-diagonal matrices) that facilitates its evaluation at each~sensor~set.

\end{abstract}

% % % % % % % % % % % % % % % % % % % % % % % % % % % % % % % % % % % % % % % % % % % % % % % % % % % % % % %
% % Section % % % %
% % % % % % % % % %
\section{Introduction}\label{sec:Intro}
% % % % % % % % % % % % % % % % % % % % % % % % % % % % % % % % % % % % % % % % % % % % % % % % % % % % % % % 

% \paragraph*{Literature review}  

Adversarial target tracking and capturing~\cite{masazade2012sparsity, karnad2015robot}, robotic navigation and autonomous construction~\cite{vitus2011sensor}, active perception and simultaneous localization and mapping (SLAM) \cite{kaess2008isam} are only a few of the challenging information gathering problems that benefit from the monitoring capabilities of sensor networks~\cite{rowaihy2007survey}. These problems are challenging because:
\begin{itemize}
\item they involve systems whose evolution is largely unknown, modeled either as a stochastic process, such as a Gaussian process~\cite{karlin2014first}, or as linear or nonlinear system corrupted with process noise~\cite{masazade2012sparsity},

\item they involve nonlinear sensors (e.g., cameras, radios) corrupted with noise~\cite{kailath2000linear}, 

\item they involve systems that change over time~\cite{nowak2004estimating}, and as a result, necessitate both spatial and temporal deployment of sensors in the environment, increasing the total number of needed sensors, and at~the~same~time,

\item they involve operational constraints, such as limited communication bandwidth and battery life, which limit the number of sensors that can simultaneously be active in the information gathering process~\cite{hero2011sensor}.
\end{itemize}

Due to these challenges, we focus on the following question: ``How do we select, at each time, only a few of the available sensors so as to monitor effectively a system despite the above challenges?''  In particular, we focus on the following sensor scheduling problem:

\newtheorem{problem1}{Problem}
\begin{problem1}
Consider a stochastic process, whose realization at time $t$ is denoted by $x(t)$ and a set of $m$ sensors, whose measurements are nonlinear functions of $x(t)$, evaluated at a fixed set of $K$ measurement times $t_1,t_2,\ldots, t_K$. In addition, suppose that at each $t_k$ a set of at most $s_k \leq m$ sensors can be used. Select the sensor sets so that the error of the corresponding minimum mean square error estimator of $(x(t_1),x(t_2),\ldots, x(t_K))$ is minimal among all possible sensor sets.
\end{problem1}

% We consider relevant papers on sensor scheduling both for batch state estimation and Kalman filtering.  Specifically, 
There are two classes of sensor scheduling algorithms, that trade-off between the estimation accuracy of the batch state vector and their time complexity~\cite{2016arXiv160807533T}: those used for Kalman filtering, and those for batch state estimation. The most relevant papers on batch state estimation are~\cite{2016arXiv160807533T} and~\cite{Roy2016369}.  However, both of these papers focus on linear systems and measurements. The most relevant papers for Kalman filtering consider algorithms that use: myopic heuristics~\cite{shamaiah2010greedy}, tree pruning~\cite{vitus2012efficient}, convex optimization~\cite{joshi2009sensor,ny2011scheduling, shen2014sensor,liu2014optimal}, quadratic programming~\cite{mo2011sensor}, Monte Carlo methods \cite{he2006sensor}, or submodular function maximization \cite{zhang2015Sensor,jawaid2015submodularity}.  However, these papers focus similarly on linear or nonlinear systems and measurements, and do not consider unknown dynamics.

%\paragraph*{Main contributions}
\myparagraph{Main contributions}
\begin{enumerate}[label=\arabic*),nosep,leftmargin=*]
\item We prove that Problem 1 is NP-hard.
\item We prove that the best approximation factor one can achieve in polynomial time for Problem 1 is $1/e$.
\item We provide Algorithm~\ref{alg:1} for Problem 1 that:
\begin{itemize}[nosep,leftmargin=*]
\item for all stochastic processes and nonlinear measurements, achieves a solution that is up to a multiplicative factor $1/2$ from the optimal solution with time complexity that is only linear in the planning horizon $K$. This is important, since it implies that Algorithm~\ref{alg:1} offers a significant computational advantage with negligible loss in performance over the polynomial-time algorithm that achieves the best approximation factor of $1/e$,
\item for important classes of Gaussian processes, and nonlinear measurements corrupted with Gaussian noise, has the same time complexity as even state-of-the-art algorithms for linear systems and measurements. For example, for Gaussian process such as those in target tracking, or those generated by linear or nonlinear systems corrupted with Gaussian noise, Algorithm~\ref{alg:1} has the same time complexity as the batch state estimation algorithm in~\cite{2016arXiv160807533T}, and lower than the Kalman filter scheduling algorithms in~\cite{joshi2009sensor,liu2014optimal}.
\end{itemize}
Therefore, Algorithm~\ref{alg:general} can enjoy both the estimation accuracy of the batch state scheduling algorithms (compared to the Kalman filtering approach, that only approximates the batch state estimation error with an upper bound~\cite{2016arXiv160807533T}) and, surprisingly, even the low time complexity of the Kalman filtering scheduling algorithms for linear systems.
\end{enumerate}

%\paragraph*{Technical contributions}
\myparagraph{Technical contributions}
% %\footnote{This observation in~\cite{jawaid2015submodularity} disproves previous results in the literature~\cite{huber2009distributed}.}

\begin{enumerate}[label=\arabic*),nosep,leftmargin=*,wide]
\item \textit{Supermodularity in Problem 1}: We achieve the approximation performance of Algorithm~\ref{alg:1}, and the linear dependence of its time complexity on the planning horizon, by proving that our estimation metric is a supermodular function in the choice of the utilized sensors.  This is important, since this is in contrast to the case of multi-step Kalman filtering for linear systems and measurements, where the corresponding estimation metric is neither supermodular nor submodular~\cite{zhang2015Sensor}~\cite{jawaid2015submodularity}. Moreover, our submodularity result cannot be reduced to the batch estimation problems in~\cite{ko1995exact,krause2008near}.  The reasons are twofold: i) we consider sensors that measure nonlinear combinations of the elements of $x(t)$, in contrast to~\cite{ko1995exact,krause2008near}, where each sensor measures directly only one element of $x(t)$; ii) our estimation metric is relevant to monitoring dynamical systems and different to the submodular entropy metric and information gain considered in~\cite{ko1995exact} and~\cite{krause2008near}, respectively.

\item \textit{Sparsity in Problem 1}: We achieve the reduced time complexity of Algorithm~\ref{alg:1} for Gaussian processes by identifying a sparsity pattern in our estimation metric.  Specifically, in this case the time complexity of each evaluation of our metric is decided by the sparsity pattern of either the covariance of $(x(t_1),x(t_2), \ldots, x(t_K))$, or the inverse of this covariance.  This is important since the two matrices are not usually sparse at the same time, even if one of them is~\cite{anderson2015batch}.  
E.g., for Gaussian processes such as those in target tracking, the first matrix is block tri-diagonal, whereas for those in SLAM, or those generated by linear or nonlinear systems corrupted with Gaussian noise, the second matrix is block tri-diagonal.
\end{enumerate}

\myparagraph{Notation} We denote the set of natural numbers $\{1,2,\ldots\}$ by $\mathbb{N}$, the set of real numbers by  $\mathbb{R}$,  and the set $\{1, 2, \ldots, n\}$ by $[n]$ ($n \in \mathbb{N}$).  The set of real numbers between $0$ and $1$ is denoted by $[0,1]$, and the empty set by $\emptyset$.  Given a set $\mathcal{X}$, $|\mathcal{X}|$ is its cardinality.  
In addition, for $n \in \mathbb{N}$, $\mathcal{X}^n$ is the $n$-times Cartesian product $\mathcal{X} \times \mathcal{X} \times\cdots \times\mathcal{X}$. 
Matrices are represented by capital letters and vectors by lower-case letters.  We write $A \in \mathcal{X}^{n_1 \times n_2}$ ($n_1, n_2 \in \mathbb{N}$) to denote a matrix of $n_1$ rows and $n_2$ columns whose elements take values in $\mathcal{X}$; $A^\top$ is its transpose, and $[A]_{ij}$ is its element at the $i$-th row and $j$-th column;
%$\|A\|_2\equiv \sqrt{A^\top A}$ is its spectral norm, and 
$\det(A)$ is its determinant.  Furthermore, if $A$ is positive definite, we write $A\succ {0}$.  In the latter case, $A^{-1}$ is its inverse.
$I$ is the identity matrix; its dimension is inferred from the context.  Similarly for the zero matrix $0$.
%--- and $e_i$, for $i \in [n]$, the $i$-th unit vector. 
The $\equiv$ denotes equivalence.
Moreover,  for a probability space $(\Omega, \mathcal{F},\mathbb{P})$, $\Omega$ is the sample space, $\mathcal{F}$ the $\sigma$-field, and $\mathbb{P}:\mathcal{F}\mapsto [0,1]$ the function that assigns probabilities to events in $\mathcal{F}$~\cite{durrett2010probability}.  We write $x \sim \mathcal{F}$ to denote a random variable $x$ with probability distribution $\mathcal{F}$; $\mathbb{E}(x)$ is its expected value, and $\Sigma(x)$ its covariance.  $x \sim \mathcal{N}(\mu, \Sigma)$ denotes a Gaussian random variable $x$ with mean $\mu$ and covariance $\Sigma$; we equivalently write $x \sim \mathcal{N}(\mathbb{E}(x), \Sigma(x))$.  Finally, we write $x|y \sim \mathcal{G}$ to denote that $x$'s probability distribution given $y$ is $\mathcal{G}$.

%We organize the rest of the paper as follows:  In Section~\ref{sec:Prelim} we formulate Problem 1.  In Section~\ref{sec:main}, we present our two main results: first, we prove that our sensor scheduling problem is NP-hard; second, we derive our approximation algorithm, and emphasize on its time complexity.  
%Section~\ref{sec:conc} concludes the paper with our future work. 
%Most proofs are found in the Appendix ---\emph{due to
%space constraints some proofs are omitted, and can be found in the full version of
%this paper, located in the authors'~websites.}%\emph{All proofs are found in the Appendix.}

% % % % % % % % % % % % % % % % % % % % % % % % % % % % % % % % % % % % % % % % % % % % % % % % % % % % % % %
% % Section % % % %
% % % % % % % % % %
\section{Problem Formulation} \label{sec:Prelim}
% % % % % % % % % % % % % % % % % % % % % % % % % % % % % % % % % % % % % % % % % % % % % % % % % % % % % % % 
%In the following paragraphs, we present our sensor scheduling problem in stochastic processes with nonlinear observations.  To this end, we first build our system and measurement model and, then, our sensor scheduling model.
%We start in more detail with the system model: 

This section introduces the system, measurement, and scheduling models and presents the sensor scheduling problem formally.
\begin{mysystem}
We consider two cases:
\begin{itemize}
\item \emph{Continuous time model}: Consider the stochastic process (along with a probability space $(\Omega, \mathcal{F},\mathbb{P})$):
\begin{equation}\label{eq:dynamics_general}
x_{\omega}(t): \omega \in \Omega, t \geq t_0 \mapsto \mathbb{R}^n
\end{equation}
where $n \in \mathbb{N}$, $t_0$ is the initial time, and $x_{\omega}(t)$ the state vector given the sample $\omega$.

\item \emph{Discrete time model}: Consider the nonlinear discrete-time system:
\begin{equation}\label{eq:dynamics_general_discrete}
x_{k+1}=l_k(x_{1:k}), l_k \sim \mathcal{L}_k,  k \in \mathbb{N}
\end{equation}
where $x_k \in \mathbb{R}^n$ is the state vector, $x_{1:k}$ the batch vector $(x_1,x_2, \ldots, x_k)$, and $\mathcal{L}_k$ a probability distribution over functions $l_k:\mathbb{R}^{nk}\mapsto \mathbb{R}^{n}$.
\end{itemize}
\end{mysystem}

% ; merely the knowledge of $\mu(t)$ and $\Sigma(t,t')$ describes the current and future state of $x(t)$: for a current time $k \in [K]$, $x(t) \sim \mathcal{N}(\mu(t), \Sigma(x(t)))$; for a future time $t'>t$, $\mu(t')$ and $\Sigma(t,t')$~are~known.  
Because the system models~\eqref{eq:dynamics_general} and~\eqref{eq:dynamics_general_discrete} assume no characteristic structure, they are appropriate for modeling largely unknown dynamics.  For example, an instance of~\eqref{eq:dynamics_general} is the time-indexed Gaussian process system model:
\begin{equation}\label{eq:dynamics_gaussian}
x(t) \sim \mathcal{GP}(\mu(t), \Sigma(t,t')),\quad t, t' \geq t_0,
\end{equation}
where $\mu(t)$ is the mean function and $\Sigma(t,t')$ is the covariance function. Similarly, an instance of~\eqref{eq:dynamics_general_discrete} is the state-indexed Gaussian process system model:
\begin{equation}\label{eq:dynamics_gaussian_discrete}
x_{k+1}=l(x_k), \quad l \sim \mathcal{GP}(\mu(x), \Sigma(x,x')), x,x' \in \mathbb{R}^n.
\end{equation}
%where $\mu(x)$ is the mean function, and $\Sigma(x,x')$ is the covariance function.

%We introduce the measurement model: 

\begin{mymeas}
Consider $m$ nonlinear sensors that operate in discrete time:
\begin{equation}\label{eq:nonlinear_measurements}
z_{i,k} = g_i(x_k)+v_{i,k}, \quad i \in [m], k \in \mathbb{N}
\end{equation}
where for the continuous-time system in~\eqref{eq:dynamics_general} we let $x_k := x(t_k)$ at a pre-specified set of measurement times $t_1,t_2,\ldots$ and $v_{i,k}$ is the measurement noise of sensor $i$ at time $k$.
%We consider two cases:
%\begin{itemize}
%\item \emph{Continuous time model}: Let the $m$ nonlinear sensors:
%\begin{equation}\label{eq:nonlinear_measurements}
%z_i(t)= g_i(x(t))+v_i(t), i \in [m],
%\end{equation}
%where $g_i: \mathbb{R}^n \mapsto \mathbb{R}^{d_i}$  ($d_i \in \mathbb{N}$) is sensor $i$'s measurement at $t$, and $v_i(t)$ its measurement noise.

%\item \emph{Discrete time model}: Let the $m$ nonlinear sensors:
%\begin{equation}\label{eq:nonlinear_measurements_discrete}
%z_k^{(i)}= g^{(i)}(x_k)+v_k^{(i)}, i \in [m].
%\end{equation}
%\end{itemize}
%For the continuous time case, we consider $m$ nonlinear sensors:
%\begin{equation}\label{eq:nonlinear_measurements}
%z_i(t)= g_i(x(t))+v_i(t), i \in [m],
%\end{equation}
%where $g_i: \mathbb{R}^n \mapsto \mathbb{R}^{d_i}$  ($d_i \in \mathbb{N}$) is a one-time differentiable function that denotes sensor $i$'s measurement at $t$; $v_i(t)$ denotes its measurement noise.
\end{mymeas}

\begin{myassump}\label{ass:indep}
$v_{i,k}$ are independent across $i$ and $k$. In addition, $g_i$ is one-time differentiable.
\end{myassump}

%\begin{myassump}\label{ass:indep}
%$v_i(t)$ are independent across $i$ and $t$; similarly for $v_k^{(i)}$. In addition, $g_i$ and $g^{(i)}$ are one-time differentiable.
%\end{myassump}

%We now introduce the sensor scheduling model:

\begin{mysensorsch}
The $m$ sensors in~\eqref{eq:nonlinear_measurements} are used at $K$ scheduled measurement times $\{t_1,t_2,\ldots,t_K\}$. At each $k \in [K]$, only $s_k$ of the $m$ sensors are used ($s_k \leq m$), resulting in the batch measurement vector~$y_k$:
\begin{equation}\label{eq:measurements}
y_k = S_kz_k, \quad k \in [K],
\end{equation}
where $S_k$ is a sensor selection matrix, composed of sub-matrices $[S_k]_{ij}$ ($i \in [s_k]$, $j \in [m]$) such that $[S_k]_{ij}=I$ if sensor $j$ is used at time $k$, and $[S_k]_{ij}=0$ otherwise. We assume that a sensor can be used at most once at each $k$, and as a result, for each $i$ there is one $j$ such that $[S_k]_{ij}=I$ while for each $j$ there is at most one $i$ such that $[S_k]_{ij}=I$.
\end{mysensorsch}

%\begin{mysensorsch}
%In the continuous time case (\eqref{eq:dynamics_general} and~\eqref{eq:nonlinear_measurements}), the $m$ sensors  in \eqref{eq:nonlinear_measurements} are used at $K$ scheduled measurement times $\{t_1,t_2,\ldots, t_K\}$; in the discrete time case (\eqref{eq:dynamics_general_discrete} and~\eqref{eq:nonlinear_measurements_discrete}), the $m$ sensors in \eqref{eq:nonlinear_measurements_discrete} are used at each $k \in [K]$.  Without loss of generality, we continue with the continuous time case: at each $t_k$ only $s_k$ of the $m$ sensors in~\eqref{eq:nonlinear_measurements} are used ($s_k \leq m$), resulting in the batch measurement vector~$y(t_k)$:
%\begin{equation}\label{eq:measurements}
%y(t_k)= S_kz(t_k), k \in [K],
%\end{equation}
%where $z(t_k)\equiv (z_1^\top(t_k),z_2^\top(t_k),\ldots,z_m^\top(t_k))^\top$, and $S_k$ is the sensor selection matrix: it is a block matrix, composed of matrices $[S_k]_{ij}$ ($i \in [s_k]$, $j \in [m]$) such that $[S_k]_{ij}=I$ if sensor $j$ is used at $t_k$, and $[S_k]_{ij}=0$ otherwise.  We consider that each sensor can be used at most once at each $t_k$, and as a result, for each $i$ there is one $j$ such that $[S_k]_{ij}=I$ while for each $j$ there is at most one $i$ such that $[S_k]_{ij}=I$.
%\end{mysensorsch}

We now present the sensor scheduling problem formally: 
\paragraph*{Notation} For $i,j \in \mathbb{N}$, $\phi_{i:j} \equiv \left(\phi_i, \phi_{i+1}, \ldots, \phi_j\right)$. In addition, $\mathcal{S}_k\equiv \{j :  \text{there exists } i \in [s_k], [S_k]_{ij}=I\}$: $\mathcal{S}_k$ is the set of indices that correspond to utilized sensors at~$t_k$.

%[Sensor Scheduling for Batch State Estimation in Stochastic Processes with nonlinear Observations]
% Without loss of generality, only the statement for the continuous time case (\eqref{eq:dynamics_general} and~\eqref{eq:nonlinear_measurements}) follows: 
% In~mathematical notation:
\begin{myproblem}[Sensor Scheduling in Stochastic Processes with Nonlinear Observations]
\label{pr:problem}
Select at each time $k$ a subset of $s_k$ sensors, out of the $m$ sensors in~\eqref{eq:nonlinear_measurements}, to use in order to minimize the conditional entropy of $x_{1:K}$  given the measurements $y_{1:K}$:
\begin{equation*}
\begin{aligned}
& \underset{\mathcal{S}_k \subseteq [m], k \in [K]}{\text{minimize}} 
 \;  \mathbb{H}(x_{1:K}|\mathcal{S}_{1:K}) \\
&\hspace{3mm}\text{subject to} 
\quad \hspace{-1.5mm} |\mathcal{S}_k| \leq s_k, k \in [K],
\end{aligned}
\end{equation*}
where $\mathbb{H}(x_{1:K}|\mathcal{S}_{1:K})$ denotes the conditional entropy $\mathbb{H}(x_{1:K}|y_{1:K})$ of $x_{1:K}$ given the measurements $y_{1:K}$.
\end{myproblem}

%\paragraph*{Notation} We use  $\mathbb{H}\left(x_{1:K}| \mathcal{S}_{1:K}\right)$ and $\mathbb{H}(x_{1:K}|y_{1:K})$ interchangeably. 

The conditional entropy $\mathbb{H}(x_{1:K}|y_{1:K})$ captures the estimation accuracy of $x_{1:K}$ given $y_{1:K}$, as we explain in the following two propositions:

\begin{myproposition}\label{prop:mutual}
$\mathbb{H}(x_{1:K}|y_{1:K})$ is a constant factor away from the mutual information of $x_{1:K}$ and $y_{1:K}$.  In particular:
\begin{equation*}\label{eq:mutual_lin_gau}
\mathbb{H}(x_{1:K}|y_{1:K})=-\mathbb{I}(x_{1:K}; y_{1:K})+\mathbb{H}(x_{1:K}),
\end{equation*} 
where $\mathbb{I}(x_{1:K}; y_{1:K})$ is the mutual information of $x_{1:K}$ and $y_{1:K}$, and $\mathbb{H}(x_{1:K})$ is constant.
\end{myproposition}

%\begin{myremark}
%Due to Proposition~\ref{prop:mutual}, our results in this paper extend to $\mathbb{I}(x_{1:K}; y_{1:K})$.
%\end{myremark}

\begin{myproposition}\label{prop:entropy} Consider the Gaussian process~\eqref{eq:dynamics_gaussian} and suppose that the measurement noise in~\eqref{eq:nonlinear_measurements} is Gaussian, $v_{i,k} \sim \mathcal{N}(0,\Sigma(v_{i,k}))$.
$\mathbb{H}(x_{1:K}|y_{1:K})$ is a constant factor away from $\log\det(\Sigma(x^{\star}_{1:K}))$, where $\Sigma(x^{\star}_{1:K})$ is the error covariance of the minimum mean square estimator $x^{\star}_{1:K}$ of $x_{1:K}$ given the measurements $y_{1:K}$.  In particular:\footnote{We explain $x^{\star}_{1:K}$ and  $\log\det(\Sigma(x^{\star}_{1:K}))$:  $x^{\star}_{1:K}$ is the optimal estimator for $x_{1:K}$, since it minimizes among \emph{all} estimators of $x_{1:K}$ the mean square error $\mathbb{E}(\|x_{1:K}-x^{\star}_{1:K}\|_2^2)$ ($\|\cdot\|_2$ is the euclidean norm), where the expectation is taken with respect to $y_{1:K}$~\cite{kailath2000linear}. $\log\det(\Sigma(x^{\star}_{1:K}))$ is an estimation error metric related to $\|x_{1:K}-x^{\star}_{1:K}\|_2^2$, since when it is minimized, the probability that the estimation error $\|x_{1:K}-x^{\star}_{1:K}\|_2^2$ is small is maximized~\cite{2016arXiv160807533T}.}
\begin{equation*}\label{eq:entropy_lin_gau}
\mathbb{H}(x_{1:K}|y_{1:K})=\frac{\log\det(\Sigma(x^{\star}_{1:K}))}{2}+\frac{nK\log(2\pi e)}{2}.
\end{equation*}
\end{myproposition}

%We derive two formulas for $\mathbb{H}(x_{1:K}|y_{1:K})$ in Appendix~\ref{app:closed_formula}.

% % % % % % % % % % % % % % % % % % % % % % % % % % % % % % % % % % % % % % % % % % % % % % % % % % % % % % %
% % Section % % % %
\section{Main Results}\label{sec:main}
% % % % % % % % % % % % % % % % % % % % % % % % % % % % % % % % % % % % % % % % % % % % % 

We first prove that Problem 1 is NP-hard, and then derive for it a provably near-optimal approximation algorithm:

\begin{mytheorem}\label{th:np}
\textit{The problem of sensor scheduling in stochastic processes with nonlinear observations (Problem~1) is NP hard.}
\end{mytheorem}
\begin{proof}
See Appendix~\ref{app:np}. Our approach is to find an instance of Problem 1 that is equivalent to the NP-hard minimal observability problem introduced in~\cite{olshevsky2014minimal,sergio2015minimal}.
\end{proof}

Due to Theorem~\ref{th:np}, we need to appeal to approximation algorithms to obtain a solution to Problem 1 in polynomial-time. To this end, we propose an efficient near-optimal algorithm (Algorithm~\ref{alg:general} with a subroutine in Algorithm~\ref{alg:greedy_alg}) and quantify its performance and time complexity in the following theorem.

% in the following paragraphs,
%To this end, in the following paragraphs, we provide an efficient provably near-optimal approximation algorithm: 
%We propose Algorithm~\ref{alg:general} for Problem 1, where Algorithm~\ref{alg:greedy_alg} is used as a subroutine; with the following theorem, we quantify its approximation performance and time complexity.

\begin{algorithm}[tl]
\caption{Approximation algorithm for Problem 1.}
\label{alg:1}
\begin{algorithmic}
\REQUIRE  Horizon $K$, scheduling constraints $s_1, s_2, \ldots, s_K$, error metric $\mathbb{H}(x_{1:K}|\mathcal{S}_{1:K}): \mathcal{S}_k \subseteq [m], k \in [K] \mapsto \mathbb{R}$
\ENSURE Sensor sets $(\mathcal{S}_1, \mathcal{S}_2,\ldots, \mathcal{S}_K)$ that approximate the solution to Problem 1, as quantified in Theorem~\ref{th:alg_performance}
\STATE $k \leftarrow 1$, $\mathcal{S}_{1:0} \leftarrow \emptyset$
\WHILE {$k\leq K$} \STATE{ 
\vspace{-4.9mm}
	\begin{enumerate}[label=\arabic*.]
    \item Apply Algorithm~\ref{alg:greedy_alg} to
    \begin{equation}\label{eq:local_opt}
    \min_{S \subseteq [m]}\{\mathbb{H}(x_{1:K}|\mathcal{S}_{1:k-1},\mathcal{S}): |\mathcal{S}|\leq s_k\}
    \end{equation}
    \item Denote by $\mathcal{S}_k$ the solution Algorithm~\ref{alg:greedy_alg} returns
	\item $\mathcal{S}_{1:k} \leftarrow (\mathcal{S}_{1:k-1},\mathcal{S}_k)$
	\item $k \leftarrow k+1$
	\end{enumerate}	
	\vspace{-1.5mm}
	}
\ENDWHILE
\end{algorithmic} \label{alg:general}
\end{algorithm}

\begin{mytheorem}\label{th:alg_performance}
The theorem has two parts: 
%(that hold for the continuous time case~\eqref{eq:dynamics_general} and~\eqref{eq:nonlinear_measurements}, and the discrete time case~\eqref{eq:dynamics_general_discrete} and~\eqref{eq:nonlinear_measurements_discrete}):
% , where $\mathcal{S}_{1:K}\equiv(\mathcal{S}_1,\mathcal{S}_2, \ldots, \mathcal{S}_K)$,
% 
\begin{enumerate}[label=\arabic*)]
\item \emph{Approximation performance of Algorithm~\ref{alg:general}:} Algorithm~\ref{alg:general} returns sensors sets $\mathcal{S}_1, \mathcal{S}_2, \ldots, \mathcal{S}_K$ that:  
\begin{enumerate}[label=\roman*.]
\item satisfy all the feasibility constraints of Problem 1: $|\mathcal{S}_k|\leq s_k, k \in [K]$
\item achieve an error $\mathbb{H}(x_{1:K}|\mathcal{S}_{1:K})$ such that:
\begin{equation}
\frac{\mathbb{H}(x_{1:K}|\mathcal{S}_{1:K})-OPT}{MAX-OPT}\leq \frac{1}{2}, \label{ineq:opt_guar}
\end{equation}
\end{enumerate}
where $OPT$ is the optimal cost of Problem 1, and $MAX \equiv \max_{\mathcal{S}'_{1:K}}\mathbb{H}(x_{1:K}|\mathcal{S}'_{1:K})$ is the maximum (worst) cost in Problem 1.

\item \emph{Time complexity of Algorithm~\ref{alg:general}:} 
Algorithm~\ref{alg:general} has time complexity $O(\sum_{k=1}^K s_k^2T)$, where $T$ is the time complexity of evaluating $\mathbb{H}(x_{1:K}|\mathcal{S}_{1:K}'): \mathcal{S}_k' \subseteq [m],$ $k \in [K]\mapsto \mathbb{R}$ at an $\mathcal{S}_{1:K}'$.  
\end{enumerate}
\end{mytheorem}

In the following paragraphs, we discuss Algorithm~\ref{alg:general}'s approximation quality and time complexity and fully characterize the latter in Theorem~\ref{th:time_compl_evaluation} and Corollary \ref{cor:time_compl_gaussian} for Gaussian processes and Gaussian measurement noise. 

%In the following paragraphs, we comment on Algorithm~\ref{alg:general}'s approximation quality and time complexity.  In addition, in Theorem~\ref{th:time_compl_evaluation} and Corollary \ref{cor:time_compl_gaussian} we fully characterize the time complexity of Algorithm~\ref{alg:general} for Gaussian stochastic processes and measurement noises.

\paragraph*{Supermodularity and monotonicity of $\mathbb{H}(x_{1:K}|y_{1:K})$}
We state two properties of $\mathbb{H}(x_{1:K}|y_{1:K})$ that are used to prove Theorem~\ref{th:alg_performance}. In particular, we show that $\mathbb{H}(x_{1:K}|y_{1:K})$ is a non-increasing and supermodular function with respect to the sequence of selected sensors. Then, Theorem~\ref{th:alg_performance} follows by combining these two results with results on submodular functions maximization over matroid constraints~\cite{fisher1978analysis}. These derivations are presented in Appendix~\ref{app:proof_of_theorem}.

% as we comment in the next paragraph, we prove that 
\paragraph*{Approximation quality of Algorithm~\ref{alg:general}}
Theorem~\ref{th:alg_performance}  quantifies the worst-case performance of Algorithm~\ref{alg:general} across all values of Problem 1's parameters.  The reason is that the right-hand side of~\eqref{ineq:opt_guar} is constant.  In particular,~\eqref{ineq:opt_guar} guarantees that for any instance of Problem 1, the distance of the approximate cost $\mathbb{H}(x_{1:K}|\mathcal{S}_{1:K})$ from $OPT$ is at most $1/2$ the distance of the worst (maximum) cost $MAX$ from $OPT$. This approximation factor is close to the optimal approximation factor $1/e \cong .38$ one can achieve in the worst-case for Problem 1 in polynomial time~\cite{vondrak2010submodularity}; the reason is twofold: first, Problem 1 involves the minimization of a non-increasing and supermodular function~\cite{Nemhauser:1988:ICO:42805}, and second, as we proved in Theorem \ref{th:np}, Problem 1 is in the worst-case equivalent to the minimal observability problem introduced in~\cite{olshevsky2014minimal}, which cannot be approximated in polynomial time with a better factor than the $1/e$~\cite{Feige:1998:TLN:285055.285059}.

% , which is excessively~high
\begin{myremark}\label{rem:best_apprx}
We can improve the $1/2$ approximation factor of Algorithm \ref{alg:general} to $1/e$ by utilizing the algorithm introduced in \cite{vondrak2008optimal}.  However, this algorithm has time complexity $O((nK)^{11}T)$, where $T$ is the time complexity of evaluating $\mathbb{H}(x_{1:K}|\mathcal{S}_{1:K}'): \mathcal{S}_k' \subseteq [m],$ $k \in [K]\mapsto \mathbb{R}$ at an $\mathcal{S}_{1:K}'$.
\end{myremark}

\paragraph*{Time complexity of Algorithm~\ref{alg:general}}
Algorithm~\ref{alg:general}'s time complexity is broken down into two parts: a) the number of evaluations of $\mathbb{H}(x_{1:K}|y_{1:K})$ required by the algorithm; b)~the time complexity of each such evaluation.  In more detail:

\paragraph{Number of evaluations of $\mathbb{H}(x_{1:K}|y_{1:K})$ required by Algorithm~\ref{alg:general}}
Algorithm~\ref{alg:general} requires at most $s_k^2$ evaluations of $\mathbb{H}(x_{1:K}|y_{1:K})$ at each $k\in [K]$.    
Therefore, Algorithm~\ref{alg:general} achieves a time complexity that is only linear in $K$ with respect to the number of evaluations of $\mathbb{H}(x_{1:K}|y_{1:K})$; the reason is that $\sum_{k=1}^K s_k^2 \leq \max_{k\in [K]}(s_k^2)K$.  This is in contrast to the algorithm in Remark \ref{rem:best_apprx}, that obtains the best approximation factor $1/e$, whose time complexity is of the order $O((nK)^{11})$ with respect to the number of evaluations of $\mathbb{H}(x_{1:K}|y_{1:K})$.\footnote{We can also speed up Algorithm~\ref{alg:general} by implementing in Algorithm~\ref{alg:greedy_alg} the method of lazy evaluations~\cite{minoux1978accelerated}: this method avoids in Step 2 of Algorithm~\ref{alg:greedy_alg} the computation of $\rho_i(\mathcal{S}^{t-1})$ for unnecessary choices of $i$.}

\begin{algorithm}[tl]
\caption{Single step greedy algorithm (subroutine in Algorithm~\ref{alg:general}).}
\begin{algorithmic}
\REQUIRE Current iteration $k$, selected sensor sets $(\mathcal{S}_1, \mathcal{S}_2,\ldots,$ $\mathcal{S}_{k-1})$ up to the current iteration, constraint $s_k$, error metric $\mathbb{H}(x_{1:K}|\mathcal{S}_{1:K}): \mathcal{S}_k \subseteq [m], k \in [K] \mapsto \mathbb{R}$
\ENSURE Sensor set $\mathcal{S}_k$ that approximates the solution to Problem 1 at time $k$
\vspace{0.25mm}
\STATE $\mathcal{S}^0\leftarrow\emptyset$, $\mathcal{X}^0\leftarrow [m]$, and $t\leftarrow 1$
\STATE \textbf{Iteration t:}
\begin{enumerate}[label=\arabic*.]
\item If $\mathcal{X}^{t-1}=\emptyset$, \textbf{return} $\mathcal{S}^{t-1}$
\item Select $i(t)\in \mathcal{X}^{t-1}$ for which $\rho_{i(t)}(\mathcal{S}^{t-1})=\max_{i \in \mathcal{X}^{t-1}}\rho_i(\mathcal{S}^{t-1})$, with ties settled arbitrarily, where:
%\[
%\!\!\negqquad\rho_i(\mathcal{S}^{t-1}) \!\equiv\! \mathbb{H}(x_{1:K}|\mathcal{S}_{1:k-1},\! \mathcal{S}^{t-1}\!)-\mathbb{H}(x_{1:K}|\mathcal{S}_{1:k-1},\! \mathcal{S}^{t-1}\cup\{i\}\!)
%\]
\begin{eqnarray*}
\rho_i(\mathcal{S}^{t-1}) &\equiv &\mathbb{H}(x_{1:K}|\mathcal{S}_{1:k-1}, \mathcal{S}^{t-1})-\\
&&\qquad\mathbb{H}(x_{1:K}|\mathcal{S}_{1:k-1}, \mathcal{S}^{t-1}\cup\{i\}) 
\end{eqnarray*}
\item[{3.a.}] If $|\mathcal{S}^{t-1}\cup \{i(t)\}| > s_k$, $\mathcal{X}^{t-1}\leftarrow \mathcal{X}^{t-1}\setminus\{i(t)\}$, and go to Step 1
\item[{3.b.}] If $|\mathcal{S}^{t-1}\cup \{i(t)\}| \leq s_k$, $\mathcal{S}^t\leftarrow \mathcal{S}^{t-1}\cup \{i(t)\}$ and $\mathcal{X}^{t}\leftarrow \mathcal{X}^{t-1}\setminus\{i(t)\}$
\item[{4.}] $t\leftarrow t+1$ and continue
\end{enumerate}
\end{algorithmic} \label{alg:greedy_alg}
\end{algorithm}

\paragraph{Time complexity of each evaluation of $\mathbb{H}(x_{1:K}|y_{1:K})$} This time complexity depends on the properties of both the stochastic process~\eqref{eq:dynamics_general} (similarly,~\eqref{eq:dynamics_general_discrete}) and the measurement noise $v_{i,k}$ in~\eqref{eq:nonlinear_measurements}.  For the case of Gaussian stochastic processes and measurement noises:

\begin{mytheorem}\label{th:time_compl_evaluation}
Consider the Gaussian process model~\eqref{eq:dynamics_gaussian} and suppose that the measurement noise is Guassian: $v_{i,k} \sim \mathcal{N}(0,\Sigma(v_{i,k}))$ such that $\Sigma(v_{i,k})\succ 0$. The time complexity of evaluating $\mathbb{H}(x_{1:K}|y_{1:K})$ depends on the sparsity pattern of $\Sigma(x_{1:K})$ and $\Sigma(x_{1:K})^{-1}$ as follows.
\begin{itemize}
\item Each evaluation of $\mathbb{H}(x_{1:K}|y_{1:K})$ has time complexity $O(n^{2.4}K)$, when \emph{either $\Sigma(x_{1:K})$ or $\Sigma(x_{1:K})^{-1}$} is exactly sparse (that is, block tri-diagonal).
\item Each evaluation of $\mathbb{H}(x_{1:K}|y_{1:K})$ has time complexity $O(n^{2.4}K^{2.4})$, when \emph{both $\Sigma(x_{1:K})$ and $\Sigma(x_{1:K})^{-1}$} are dense.
\end{itemize}
\end{mytheorem}

Theorem~\ref{th:time_compl_evaluation} implies that when $\Sigma(x_{1:K})$ or $\Sigma(x_{1:K})^{-1}$ is exactly sparse, the time complexity of each evaluation of $\mathbb{H}(x_{1:K}|y_{1:K})$ is only linear in $K$.  This is important because $\Sigma(x_{1:K})$ or $\Sigma(x_{1:K})^{-1}$ is exactly sparse for several applications and system models~\cite{dellaert2006square}.  For example, in adversarial target tracking applications, where the target wants to avoid capture and randomizes its motion in the environment (by un-correlating its movements), $\Sigma(x_{1:K})$ can be considered tri-diagonal (since this implies $x(t_k)$ and $x(t_{k'})$ are uncorrelated for $|k-k'|>2$).  Similarly, in SLAM, or in system models where the Gaussian process in~\eqref{eq:dynamics_gaussian} is generated by a linear or nonlinear system corrupted with Gaussian noise, $\Sigma(x_{1:K})^{-1}$ is block tri-diagonal \cite{anderson2015batch}.  In particular, for linear systems, $\Sigma(x_{1:K})^{-1}$ is block tri-diagonal~\cite[Section 3.1]{anderson2015batch}, and for nonlinear systems, $\Sigma(x_{1:K})^{-1}$ is efficiently approximated by a block tri-diagonal matrix as follows: for each $k$, before the $k$-th iteration of Step 1 in Algorithm~\ref{alg:general}, we first compute $\tilde{\mu}_{1:K}$ given $y_{1:(k-1)}$ up to $k$.  This step has complexity $O(n^{2.4}K)$ when $\Sigma(x_{1:K})^{-1}$ is sparse~\cite[Eq.~(5)]{anderson2015batch}~\cite[Section 3.8]{quarteroni2010numerical}, and it does not increase the total time complexity of Algorithm~\ref{alg:general}.  Then, we continue as in~\cite[Section 3.2]{anderson2015batch}.

\paragraph*{Sparsity in $\mathbb{H}(x_{1:K}|y_{1:K})$}
We state the two properties of $\mathbb{H}(x_{1:K}|y_{1:K})$  that result to Theorem~\ref{th:time_compl_evaluation}.  In particular, we prove that $\mathbb{H}(x_{1:K}|y_{1:K})$ is expressed in closed form with two different formulas such that the time complexity for the evaluation of $\mathbb{H}(x_{1:K}|y_{1:K})$ using the first formula is decided by the sparsity pattern of $\Sigma(x_{1:K})$, whereas using the second formula is decided by the sparsity pattern of $\Sigma(x_{1:K})^{-1}$.  
%This is important since $\Sigma(x_{1:K})$ and $\Sigma(x_{1:K})^{-1}$ are not necessarily simultaneously sparse, even if one of them is~\cite{anderson2015batch}.  
The reason for this dependence is that the rest of the matrices in these formulas are sparser than $\Sigma(x_{1:K})$ or $\Sigma(x_{1:K})^{-1}$; in particular, they are block diagonal.

% After Theorem~\ref{th:time_compl_evaluation}, t
The full characterization of Algorithm~\ref{alg:general}'s time complexity for Gaussian processes and Gaussian measurement noises follows.

\begin{mycorollary}\label{cor:time_compl_gaussian}
Consider the Gaussian process model~\eqref{eq:dynamics_gaussian} and suppose that the measurement noise is Gaussian: $v_{i,k} \sim \mathcal{N}(0,\Sigma(v_{i,k}))$ such that $\Sigma(v_{i,k})\succ 0$.  The time complexity of Algorithm~\ref{alg:general} depends on the sparsity pattern of $\Sigma(x_{1:K})$ and $\Sigma(x_{1:K})^{-1}$ as follows.
\begin{itemize}
\item Algorithm~\ref{alg:general} has time complexity $O(n^{2.4}K\sum_{k=1}^K s_k^2)$, when \emph{either $\Sigma(x_{1:K})$ or $\Sigma(x_{1:K})^{-1}$} is exactly sparse (that is, block tri-diagonal).
\item Algorithm~\ref{alg:general} has time complexity $O(n^{2.4}K^{2.4}\sum_{k=1}^K s_k^2)$, when \emph{both $\Sigma(x_{1:K})$ and $\Sigma(x_{1:K})^{-1}$} are dense.
\end{itemize}
\end{mycorollary}

% Therefore, if in~\cite{2016arXiv160807533T} we consider Gaussian process and measurement noise,  
\paragraph*{Comparison of Algorithm~\ref{alg:general}'s time complexity for Gaussian processes and Gaussian measurement noises, per Corollary \ref{cor:time_compl_gaussian}, to that of existing scheduling algorithms}
The most relevant algorithm to Algorithm~\ref{alg:general} is the one provided in~\cite{2016arXiv160807533T}, where linear systems with additive process noise and measurement noises with \emph{any} distribution are assumed. Algorithm~\ref{alg:general} generalizes~\cite{2016arXiv160807533T} from linear systems and measurements to Gaussian processes and nonlinear measurements.  At the same time, it achieves the same time complexity as the algorithm in~\cite{2016arXiv160807533T} when $\Sigma(x_{1:K})$ or $\Sigma(x_{1:K})^{-1}$ is exactly sparse.  This is important since the algorithm in~\cite{2016arXiv160807533T} has time complexity lower than the-state-of-the-art batch estimation sensor scheduling algorithms, such as the algorithm proposed in \cite{Roy2016369}, and similar to that of the state of the art Kalman filter scheduling algorithms, such as those proposed in \cite{joshi2009sensor,liu2014optimal,jawaid2015submodularity} (in particular, lower for large $K$). 
%Thus, Algorithm~\ref{alg:general} enjoys both the accuracy of batch estimation scheduling algorithms (compared to the Kalman filtering approach, that only approximates the batch state estimation error with an upper bound~\cite{2016arXiv160807533T}) and, surprisingly, even the low time complexity of the Kalman filter scheduling algorithms for linear systems.

%% % % % % % % % % % % % % % % % % % % % % % % % % % % % % % % % % % % % % % % % % % % % % % % % % % % % % % %
\section{Conclusion}\label{sec:conc}
%% % % % % % % % % % % % % % % % % % % % % % % % % % % % % % % % % % % % % % % % % % % % % 

In this paper, we proposed Algorithm~\ref{alg:general} for the NP-hard problem of sensor scheduling for stochastic process estimation. Exploiting the supermodularity and monotonicity of conditional entropy, we proved that the algorithm has an approximation factor $1/2$ and linear complexity in the scheduling horizon. It achieves both the accuracy of batch estimation scheduling algorithms and, surprisingly, when the information structure of the problem is sparse, the low time complexity of Kalman filter scheduling algorithms for linear systems. This is the case, for example, in applications such as SLAM and target tracking, and for processes generated by linear or nonlinear systems corrupted with Gaussian noise. Future work will focus on an event-triggered version of the scheduling problem, in which the measurement times are decided online based on the available measurements, and on a decentralized version, in which information is exchanged only among neighboring~sensors.

\section*{Appendix A: Proof of Proposition~\ref{prop:entropy}}\label{appendix:cond_entropy}

\begin{proof}
We first show that the conditional probability distribution of  $x_{1:K}$ given $y_{1:K}$ is Gaussian with covariance $\Sigma(x^{\star}_{1:K})$, and then apply the following lemma: 
\begin{mylemma}[Ref.~\cite{cover2012elements}]\label{lem:entropy_gaussian}
Let $x \sim \mathcal{N}(\mu, \Sigma)$ and $x \in \mathbb{R}^m$:
\[
\mathbb{H}(x)=\frac{1}{2}\log[(2\pi e)^m\det(\Sigma)].
\]
\end{mylemma}

Specifically, due to Assumption~\ref{ass:indep}, $(x_{1:K},y_{1:K})$ are jointly Gaussian.  This has a twofold implication: first, the minimum mean square estimator of $x_{1:K}$ given $y_{1:K}$ is linear in $y_{1:K}$~\cite[Proposition E.2]{bertsekas2005dynamic}; second, the conditional probability distribution of $x_{1:K}$ given $y_{1:K}$ is Gaussian~\cite{venkatesh2012theory}, with covariance $\Sigma(x^{\star}_{1:K})$.  Therefore, due to~\cite[Proposition E.3]{bertsekas2005dynamic}, this is also the covariance of the minimum mean square estimator of $x_{1:K}$ given $y_{1:K}$. As a result, due to Lemma~\ref{lem:entropy_gaussian}:
\begin{align}\label{eq:entropy_closed_app}
\mathbb{H}(x_{1:K}|y_{1:K}) &= \mathbb{E}_{y_{1:K}=y_{1:K}'}\left(\mathbb{H}(x_{1:K}|y_{1:K}=y_{1:K}')\right)\nonumber\\
&= \mathbb{E}_{y_{1:K}=y_{1:K}'}\left(\frac{1}{2}\log[(2\pi e)^{nK}\det(\Sigma(x^{\star}_{1:K}))\right)\nonumber\\
&= \frac{nK\log(2\pi e)+\log\det(\Sigma(x^{\star}_{1:K}))}{2}.
\end{align} 
We derive a formula for $\Sigma(x^{\star}_{1:K})$ in the proof of Lemma~\ref{lem:entropy_closed_inv_sigma}.~\qedhere
%where all the expectations are taken with respect to $y_{1:K}$, (according to the definition of the conditional entropy $\mathbb{H}(x_{1:K}|y_{1:K})$~\cite{cover2012elements}).
\end{proof}

\section*{Appendix B: Proof of Theorem~\ref{th:np}}\label{app:np}

\begin{proof}
We present for the discrete time case~\eqref{eq:dynamics_general_discrete} an instance of Problem 1 that is equivalent to the NP-hard minimal observability problem introduced in~\cite{olshevsky2014minimal, sergio2015minimal}, that is defined as follows (the proof for the continuous time case is similar):

\paragraph*{Definition (Minimal Observability Problem)}
\textit{Consider the linear time-invariant system:
\begin{align}\label{eq:min_obs_syst}
\begin{split}
\dot{x}(t)&= A x(t), \\
y_i(t) &= r_{i} e_i^\top x(t), i \in [n] 
\end{split}
\end{align}
where $e_i$ is the vector with the $i$-th entry equal to $1$ and the rest equal to $0$, and $r_{i}$ is either zero or one; the \emph{minimal observability problem} follows:
\begin{equation}
\begin{aligned}
& {\text{select}} 
 \; \qquad r_1, r_2, \ldots, r_n \\
&\text{such that } \hspace{2mm}r_1+ r_2+ \ldots+ r_n \leq r,\\
& \qquad\qquad \hspace{1.5mm}~\eqref{eq:min_obs_syst} \text{ is observable}.
\end{aligned}
\end{equation}
}

The minimal observability problem is NP-hard when $A$ is chosen as in the proof of Theorem 1 of~\cite{olshevsky2014minimal}, and $r \leq n$.  We denote this $A$ by $A_{NP-h}$.

Problem 1 is equivalent to the NP-hard minimal observability problem for the following instance: $K=1$, $x(t_0)\sim \mathcal{N}(c,0)$, where $c\in \mathbb{R}^n$ is an unknown constant, $\mu(t)=e^{A_{NP-h}(t-t_0)} x(t_0)$, $\Sigma(t,t')=0$, $m=n$, $g_{i}(x(t))=e_i^\top x(t)$, zero measurement noise, and $s_1=r$.  This observation concludes the proof.
\end{proof}

\section*{Appendix C: Proof of Theorem~\ref{th:alg_performance}} \label{app:proof_of_theorem}

\begin{proof}
We first prove that $\mathbb{H}(x_{1:K}|\mathcal{S}_{1:K})$ is a non-increasing and supermodular function in the choice of the~sensors.  Then, we prove Theorem~\ref{th:alg_performance} by combining these two results and results on the maximization of submodular functions over matroid constraints~\cite{fisher1978analysis}.
 
\paragraph*{Notation} Given $K$ disjoint finite sets $\mathcal{E}_1,\mathcal{E}_2, \ldots,\mathcal{E}_K$ and $A_{i}, B_{i} \in \mathcal{E}_{i}$, we write $A_{1:K} \preceq B_{1:K}$ to denote that for all $i \in [K]$, $A_i\subseteq B_i$ ($A_i$ is a subset of $B_i$).  Moreover, we denote that $A_{i} \in \mathcal{E}_{i}$ for all $i \in [K]$ by $A_{1:K} \in \mathcal{E}_{1:K}$.  In addition, given $A_{1:K}, B_{1:K} \in \mathcal{E}_{1:K}$, we write $A_{1:K} \uplus B_{1:K}$ to denote that for all $i \in [K]$, $A_i\cup B_i$ ($A_i$ union $B_i$).

% The definition of a non-increasing function follows: 
\begin{mydef}
\textit{Consider $K$ disjoint finite sets $\mathcal{E}_1,\mathcal{E}_2, \ldots,\mathcal{E}_K$.  A function $h: \mathcal{E}_{1:K} \mapsto \mathbb{R}$ is \emph{non-decreasing} if and only if for all $A,B \in \mathcal{E}_{1:K}$ such that $A \preceq B$,
$
h(A)\leq h(B);
$
$h: \mathcal{E}_{1:K}\mapsto \mathbb{R}$ is \emph{non-increasing} if $-h$ is non-decreasing.}
\end{mydef}

\begin{myproposition}\label{prop:entropy_monoton}
For any finite $K \in\mathbb{N}$, consider  $K$ distinct copies of $[m]$, denoted by $\mathcal{R}_1, \mathcal{R}_2, \ldots, \mathcal{R}_K$.  The estimation error metric
$
\mathbb{H}(x_{1:K}|\mathcal{S}_{1:K}): \mathcal{R}_{1:K} \mapsto \mathbb{R}
$
is a non-increasing function in the choice of the~sensors  $\mathcal{S}_{1:K}$.
\end{myproposition}
\begin{proof} 
Consider $A, B\in \mathcal{R}_{1:K}$ such that $A \preceq B$, and denote by $B\setminus A\equiv \{i|i\in B, i\notin A\}$:
$\mathbb{H}(x_{1:K}| B)= \mathbb{H}(x_{1:K}| A, B\setminus A)\leq \mathbb{H}(x_{1:K}| A)$ since conditioning can either keep constant or decrease the entropy~\cite{cover2012elements}.
\end{proof}

%\paragraph*{Notation} Given $K$ disjoint finite sets $\mathcal{E}_1,\mathcal{E}_2,\ldots,\mathcal{E}_K$ and $A_{1:K}, B_{1:K} \in \mathcal{E}_{1:K}$, we write $A_{1:K} \uplus B_{1:K}$ to denote that for all $i \in [K]$, $A_i\cup B_i$ ($A_i$ union $B_i$). 
 
%According to Definition~\ref{def:seq_sub}, set submodularity is a diminishing returns property: a function $h: \mathcal{E}_{1:K}\mapsto \mathbb{R}$ is set submodular if and only if for all $C \in \mathcal{E}_{1:K}$, the function $h_C: \mathcal{E}_{1:K}\mapsto \mathbb{R}$ defined for all $A \in \mathcal{E}_{1:K}$ as $h_C(A)\equiv h(A\uplus C)-h(A)$ is non-increasing.

%The definition of a supermodular set function follows:
\begin{mydef}\label{def:seq_sub}
\textit{Consider $K$ disjoint finite sets $\mathcal{E}_1,\mathcal{E}_2, \ldots,$ $\mathcal{E}_K$. A function $h: \mathcal{E}_{1:K}\mapsto \mathbb{R}$ is \emph{submodular} if and only if for all $A,B,C \in \mathcal{E}_{1:K}$ such that $A \preceq B$,
$
h(A\uplus C)-h(A)\geq h(B\uplus C)-h(B);
$
$h: \mathcal{E}_{1:K}\mapsto \mathbb{R}$ is \emph{supermodular} if $-h$ is submodular.}
\end{mydef}

\begin{myproposition}
\label{prop:seq_sub}
For any finite $K \in\mathbb{N}$, consider $K$ distinct copies of $[m]$, denoted by $\mathcal{R}_1, \mathcal{R}_2, \ldots, \mathcal{R}_K$; the estimation error metric
$
\mathbb{H}(x_{1:K}|\mathcal{S}_{1:K}): \mathcal{R}_{1:K} \mapsto \mathbb{R}
$
is a set supermodular function in the choice of the~sensors $\mathcal{S}_{1:K}$.
\end{myproposition}
\begin{proof}
Let $A, B, C \in \mathcal{E}_{1:K}$ such that $A \preceq B$:
\begin{align}
 \mathbb{H}(x_{1:K} | A)-&\mathbb{H}(x_{1:K} | A \uplus C)  \label{aux:0}\\
&=\mathbb{H}(x_{1:K} | A)-\mathbb{H}(x_{1:K} | A, C) \nonumber\\
&=\mathbb{I}(x_{1:K}; C | A) \label{aux:1}\\
&=\mathbb{H}(C | A) - \mathbb{H}(C | x_{1:K}, A) \label{aux:2}\\
&\geq \mathbb{H}(C| B) - \mathbb{H}(C | x_{1:K},B) \label{aux:3}\\
&= \mathbb{I}(x_{1:K}; C|B) \label{aux:4} \\
&= \mathbb{H}(x_{1:K}| B) - \mathbb{H}(x_{1:K} | B, C) \label{aux:5}\\
&= \mathbb{H}(x_{1:K}| B) - \mathbb{H}(x_{1:K} | B \uplus C) \label{aux:6}.
\end{align}
Eq.~\eqref{aux:0} and~\eqref{aux:6} follow from our definition of $\uplus$.~\eqref{aux:1} and~\eqref{aux:2},~\eqref{aux:3} and~\eqref{aux:4}, and~\eqref{aux:4} and~\eqref{aux:5} hold due to the definition of mutual information~\cite{cover2012elements}.~\eqref{aux:3} follows from~\eqref{aux:2} due to two reasons: 
first, $\mathbb{H}(C | A)\geq \mathbb{H}(C | B)$, since $A \preceq B$ and conditioning can either keep constant or decrease the entropy~\cite{cover2012elements}; second, $\mathbb{H}( C | x_{1:K}, A)= \mathbb{H}( C | x_{1:K}, B)$ due to the independence of the measurements given $x_{1:K}$, per Assumption~\ref{ass:indep}.
\end{proof}

%Proposition~\ref{prop:seq_sub} implies that as we increase at each $t_k$ the number of sensors used, the marginal improvement we get on the estimation error of $x_{1:K}$ diminishes.

\begin{proof}[\textit{Proof of Part 1 of Theorem~\ref{th:alg_performance}}]
We use the next result from the literature of maximization of submodular functions over matroid constraints:

%The definition of a matroid follows:
\begin{mydef}\label{def:indep_matroid}
\textit{Consider a finite set $\mathcal{E}$ and a collection $\mathcal{C}$ of subsets of $\mathcal{E}$. $(\mathcal{E}, \mathcal{C})$ is:
\begin{itemize}
\item an \emph{independent system} if and only if:
\begin{itemize}
\item $\emptyset \in \mathcal{C}$, where $\emptyset$ denotes the empty set
\item for all $X' \subseteq X \subseteq \mathcal{E}$, if $X \in \mathcal{C}$, $X' \in \mathcal{C}$.
\end{itemize}
\item a \emph{matroid} if and only if in addition to the previous two properties:
\begin{itemize}
\item for all $ X', X \in \mathcal{C}$ where $|X'| < |X|$, there exists $x \notin X'$ and $x \in X$ such that $X' \cup \{x\} \in \mathcal{C}$.
\end{itemize}
\end{itemize}}
\end{mydef}

\begin{mylemma}[Ref.~\cite{fisher1978analysis}]\label{lem:sub_guarantees}
\textit{Consider $K$ independence systems $\{(\mathcal{E}_k,\mathcal{C}_k)\}_{k \in [K]}$, each the intersection of at most $P$ matroids,
and a submodular and non-decreasing function $h: \mathcal{E}_{1:K}\mapsto \mathbb{R}$. There exist a polynomial time greedy algorithm that returns an (approximate) solution $\mathcal{S}_{1:K}$ to: 
\begin{equation}\label{pr:general}
\begin{aligned}
& \underset{\mathcal{S}_{1:K} \preceq \mathcal{E}_{1:K}}{\text{maximize}} 
 \; \quad h(\mathcal{S}_{1:K}) \\
&\hspace{0mm}\text{subject to} \quad\hspace{1mm} \mathcal{S}_k\cap \mathcal{E}_k \in \mathcal{C}_k, k \in [K],
\end{aligned}
\end{equation}
that satisfies:  
\begin{equation}\label{ineq:opt_guar_seq_sub}
\frac{h(\mathcal{O})-h(\mathcal{S}_{1:K})}{h(\mathcal{O})-h(\emptyset)}\leq \frac{P}{1+P},
\end{equation}
where $\mathcal{O}$ is an (optimal) solution to~\eqref{pr:general}.}
%\footnote{The algorithm is omitted due to space constraints.}
\end{mylemma}
%In particular, we prove:
\begin{mylemma}\label{lemma:as_Problem_1.6}
\textit{Problem 1 is an instance of~\eqref{pr:general} with $P=1$.}
\end{mylemma}
\begin{proof}
We identify the instance of $\{\mathcal{E}_k,\mathcal{C}_k\}_{k\in [K]}$ and $h$, respectively, that translate~\eqref{pr:general} to Problem 1:  

Given $K$ distinct copies of $[m]$, denoted by $\mathcal{R}_1, \mathcal{R}_2, \ldots, \mathcal{R}_K$, first consider $\mathcal{E}_k=\mathcal{R}_k$ and $\mathcal{C}_k=$ $\{\mathcal{S}|\mathcal{S}\subseteq \mathcal{R}_k, |\mathcal{S}|\leq s_k\}$: $(\mathcal{E}_k, \mathcal{C}_k)$ satisfies the first two points in part 1 of Definition~\ref{def:indep_matroid}, and as a result is an independent system.  Moreover, by its definition, $\mathcal{S}_k\cap \mathcal{E}_k \in \mathcal{C}_k$ if and only if $|\mathcal{S}_k|\leq s_k$.

Second, for all $\mathcal{S}_{1:K} \preceq \mathcal{E}_{1:K}$, consider: \[h(\mathcal{S}_{1:K})=-\mathbb{H}(x_{1:K}|\mathcal{S}_{1:K}).
\]From Propositions~\ref{prop:entropy_monoton} and~\ref{prop:seq_sub}, $h(\mathcal{S}_{1:K})$ is set submodular and non-decreasing. 
In addition to Lemma~\ref{lemma:as_Problem_1.6}, the independence system $(\mathcal{E}_k, \mathcal{C}_k)$, where $\mathcal{E}_k=\mathcal{R}_k$ and $\mathcal{C}_k=\{\mathcal{S}|\mathcal{S}\subseteq \mathcal{R}_k, |\mathcal{S}|\leq s_k\}$, satisfies also the point in part 2 of Definition~\ref{def:indep_matroid}; thereby, it is also a matroid and as a result $P$, as in Lemma~\ref{lem:sub_guarantees}, is $1$.  
\end{proof}
This observation, along with Lemmas~\ref{lem:sub_guarantees} and~\ref{lemma:as_Problem_1.6} complete the proof of~\eqref{ineq:opt_guar}, since the adaptation  to Problem 1 of the greedy algorithm in~\cite[Theorem 4.1]{fisher1978analysis} results to Algorithm~\ref{alg:general}.
\end{proof}

\begin{proof}[\textit{Proof of Part 2 of Theorem~\ref{th:alg_performance}}]
Algorithm~\ref{alg:general} requires for each $k\in [K]$ the application of Algorithm~\ref{alg:greedy_alg} to~\eqref{eq:local_opt}.  In addition, each such application requires at most $s_k^2$ evaluations of $\mathbb{H}(x_{1:K}|y_{1:K})$. Therefore, Algorithm~\ref{alg:general} has time complexity $O(\sum_{k=1}^Ks_k^2T)$.
\end{proof}
The proof of Theorem \ref{th:alg_performance} is complete. 
\end{proof}

\section*{Appendix D:  Proof of Theorem~\ref{th:time_compl_evaluation}}\label{app:time_compl_evaluation}

\paragraph*{Notations} We introduce five notations: first, $S_{1:K}$ is the block diagonal matrix with diagonal elements the sensor selection matrices $S_1, S_2, \ldots, S_K$; second, $c(x_{1:K})$ is the batch vector $[(S_1g(x_1))^\top,(S_2g(x_2))^\top, \ldots, (S_Kg(x_K))^\top]^\top$, where $g(x_k)\equiv (g_1(x_k),g_2(x_k), \ldots, g_m(x_k))^\top$; third, $C(x_{1:K})$ is the block diagonal matrix with diagonal elements the matrices $S_1C_1, S_2C_2, \ldots, S_KC_K$, where $C_k\equiv G(x_k)$ and $G(x(t))\equiv \partial g(x(t))/\partial x(t)$; fourth, $v_k$ is the batch measurement noise vector $(v_{1,k}^\top,v_{2,k}^\top, \ldots, v_{m,k}^\top)^\top$; fifth, $\mu_{1:K} \equiv (\mu(t_1)^\top, \mu(t_2)^\top,\ldots,\mu(t_K)^\top)^\top$.

\begin{proof} 
We first derive the two formulas for $\mathbb{H}(x_{1:K}|y_{1:K})$:  the first formula is expressed in terms of $\Sigma(x_{1:K})^{-1}$, and the second formula is expressed in terms of $\Sigma(x_{1:K})$. 
 
\begin{mylemma}[Formula of $\mathbb{H}(x_{1:K}|y_{1:K})$ in terms of $\Sigma(x_{1:K})^{-1}$]\label{lem:entropy_closed_inv_sigma}
Consider the start of the $k$-th iteration in Algorithm~\ref{alg:general}.   Given the measurements $y_{1:(k-1)}$ up to $k$, $\mathbb{H}(x_{1:K}|y_{1:K})$ is given by $-T_1'+nK\log(2\pi e)/2$, where:
\begin{align*}
T_1'&\equiv \frac{1}{2}\log\det(\Xi+\Sigma(x_{1:K})^{-1})\\
\Xi&\equiv C(\tilde{\mu}_{1:K})^\top S_{1:K}\Sigma(v_{1:K})^{-1}S_{1:K}^\top C(\tilde{\mu}_{1:K})
\end{align*}
and $\tilde{\mu}_{1:K}$ is the maximum a posteriori (MAP) estimate of $x_{1:K}$ given the measurements $y_{1:(k-1)}$ up to $k$.
\end{mylemma}

\begin{proof}
Before the $k$-th iteration of Step 1 in Algorithm~\ref{alg:general}, we first compute $\tilde{\mu}_{1:K}$ given $y_{1:(k-1)}$ up to $k$.  This step has complexity $O(n^{2.4}K)$ when $\Sigma(x_{1:K})^{-1}$ is sparse~\cite[Eq.~(5)]{anderson2015batch}~\cite[Section 3.8]{quarteroni2010numerical}, and it does not increase the total time complexity of Algorithm~\ref{alg:general}.  Next, given $\tilde{\mu}_{1:K}$, we linearise our measurement model over $\tilde{\mu}_{1:K}$, and compute $C(\tilde{\mu}_{1:K})$.  Then, we continue as follows: $x_{1:K}$ and $y_{1:K}$ are jointly Gaussian:
\begin{align*}
(x_{1:K},y_{1:K})&\sim \mathcal{N}\left(\mathbb{E}(x_{1:K},y_{1:K}),\Sigma(x_{1:K},y_{1:K})\right),\\
\mathbb{E}(x_{1:K},y_{1:K})&=(\mu_{1:K},c(\mu_{1:K}))\\
\Sigma(x_{1:K},y_{1:K})&=\left[\begin{array}{cc}
\Sigma(x_{1:K})& \Sigma(x_{1:K})C(\tilde{\mu}_{1:K})^\top\\
 C(\tilde{\mu}_{1:K})\Sigma(x_{1:K})&  \Sigma(y_{1:K})
\end{array}\right],
\end{align*}
where:
\begin{align*}
\Sigma(y_{1:K})&=S_{1:K}\Sigma(v_{1:K})S_{1:K}^\top+C(\tilde{\mu}_{1:K})\Sigma(x_{1:K})C(\tilde{\mu}_{1:K})^\top.
\end{align*}

Therefore, the conditional probability distribution of $x_{1:K}$ given $y_{1:K}$ has covariance $\Sigma(x^{\star}_{1:K})$ (using our notation in Proposition~\ref{prop:entropy}) such that:
\begin{equation*}
\Sigma(x^{\star}_{1:K})= \Sigma(x_{1:K})-\Sigma(x_{1:K})C(\tilde{\mu}_{1:K})^\top \Phi C(\tilde{\mu}_{1:K})\Sigma(x_{1:K}),
\end{equation*}
where: 
\begin{equation*}
\Phi\equiv C(\tilde{\mu}_{1:K})\Sigma(x_{1:K})C(\tilde{\mu}_{1:K})^\top+S_{1:K}\Sigma(v_{1:K})S_{1:K}^\top)^{-1}.
\end{equation*}
Using the Woodbury matrix identity~\cite[Corollary 2.8.8]{bernstein2009matrix}:
\begin{equation*}
\Sigma(x^{\star}_{1:K})= (\Xi+\Sigma(x_{1:K})^{-1})^{-1},
\end{equation*}
where we also used the $(S_{1:K}\Sigma(v_{1:K})S_{1:K}^\top)^{-1}=S_{1:K}$ $\Sigma(v_{1:K})^{-1}S_{1:K}^\top$, which holds since $S_{1:K}$ and $\Sigma(v_{1:K})$ are block diagonal.  Using~\eqref{eq:entropy_closed_app} the proof is complete.
\end{proof}

\begin{myremark}
The time complexity for the evaluation of $\mathbb{H}(x_{1:K}|y_{1:K})$ using Lemma~\ref{lem:entropy_closed_inv_sigma} is decided by the sparsity of $\Sigma(x_{1:K})^{-1}$ since the rest of the matrices are block diagonal.
\end{myremark}

\begin{mylemma}[Formula of $\mathbb{H}(x_{1:K}|y_{1:K})$ in terms of $\Sigma(x_{1:K})$]\label{lem:entropy_closed_Sigma}
Consider the start of the $k$-th iteration in Algorithm~\ref{alg:general}. Given the measurements $y_{1:(k-1)}$ up to $k$, $\mathbb{H}(x_{1:K}|y_{1:K})$ is given by $\mathbb{H}(x_{1:K}|y_{1:K})=T_1-T_2+\mathbb{H}(x_{1:K})$,
where:
\begin{align}
T_1&\equiv \frac{1}{2}\sum_{k=1}^K\log[(2\pi e)^{s_k}\det(S_k\Sigma(v_k)S_k^\top)]\label{aux_3_app}\\
T_2&\equiv \frac{1}{2}\log[(2\pi e)^{\sum_{k=1}^Ks_k}\det(\Sigma(y_{1:K}))]\label{aux_4_app}\\
%\hspace{3mm}
%T_3&\equiv \log[(2\pi e)^{nK}\det(\Sigma(x_{1:K}))],
\Sigma(y_{1:K})&=S_{1:K}\Sigma(v_{1:K})S_{1:K}^\top+C(\tilde{\mu}_{1:K})\Sigma(x_{1:K})C(\tilde{\mu}_{1:K})^\top,\nonumber
\end{align}
and $\tilde{\mu}_{1:K}$ is the maximum a posteriori (MAP) estimate of $x_{1:K}$ given the measurements $y_{1:(k-1)}$ up to $k$.
%and $\Sigma(y_{1:K})=S_{1:K}\Sigma(v_{1:K})S_{1:K}^\top+C(\mu_{1:K})\Sigma(x_{1:K})$ $C(\mu_{1:K})^\top$, where $\mu_{1:K}\equiv (\mu(t_1)^\top, \mu(t_2)^\top,\ldots,\mu(t_K)^\top)^\top$
\end{mylemma}
\begin{proof}
Before the $k$-th iteration of Step 1 in Algorithm~\ref{alg:general}, we first compute $\tilde{\mu}_{1:K}$ given $y_{1:(k-1)}$ up to $k$.  This step has complexity $O(n^{2.4}K)$ when: a) $\Sigma(x_{1:K})$ is sparse~\cite[Eq.~(5) after multiplying its both sides with $\Sigma(x_{1:K})$]{anderson2015batch}; b) certain invertibility conditions apply~\cite[Section 3.8]{quarteroni2010numerical}.  In this case, this step does not increase the total time complexity of Algorithm~\ref{alg:general}.  If the invertibility conditions in~\cite[Section 3.8]{quarteroni2010numerical} do not apply, the complexity of this computation is $O(n^{2.4}K^{2.4})$.  In this case, we can use $\mu_{1:K}$, instead of $\tilde{\mu}_{1:K}$, to evaluate $\mathbb{H}(x_{1:K}|y_{1:K})$, and keep the overall complexity of Algorithm~\ref{alg:general} to $O(n^{2.4}K)$. Next, given $\tilde{\mu}_{1:K}$, we linearise our measurement model over $\tilde{\mu}_{1:K}$, and compute $C(\tilde{\mu}_{1:K})$.  Then, we continue as follows: the chain rule for conditional entropies implies~\cite{cover2012elements}: $\mathbb{H}(x_{1:K}|y_{1:K})=\mathbb{H}(y_{1:K}|x_{1:K})-\mathbb{H}(y_{1:K})+\mathbb{H}(x_{1:K})$.  Thus, we derive a closed formula for $\mathbb{H}(y_{1:K}|x_{1:K})$ and $\mathbb{H}(y_{1:K})$:
% and $\mathbb{H}(x_{1:K})$, respectively:

\paragraph*{Closed form of $\mathbb{H}(y_{1:K}|x_{1:K})$} The chain rule for conditional entropies implies~\cite{cover2012elements}:
\begin{align}
\mathbb{H}(y_{1:K}|x_{1:K})&=\sum_{k=1}^K\mathbb{H}(y(t_k)|x_{1:K},y_{1:k-1})\label{aux_1_app}\\
&=\sum_{k=1}^K\mathbb{E}_{x(t_{k'})}(\mathbb{H}(y(t_k)|x(t_k)=x(t_{k'})))\label{aux_2_app}.
\end{align}
Eq.~\eqref{aux_2_app} follows from~\eqref{aux_1_app} because given $x(t_k)$ $y(t_k)$ is independent of $y_{1:k-1}$, $x_{1:(k-1)}$ and $x_{(k+1):K}$.  In addition,~\eqref{aux_3_app} follows from~\eqref{aux_2_app} because $y(t_k)|x(t_k) \sim \mathcal{N}(S_kg(x(t_k)), S_k\Sigma(v_k)S_k^\top))$ and, thus, Lemma~\ref{lem:entropy_gaussian}~applies.

\paragraph*{Closed form of $\mathbb{H}(y_{1:K})$} To this end, we derive the marginal distribution of $y_{1:K}$, denoted by $f(y_{1:K})$:
\begin{align*}
f(y_{1:K})&=\int f(y_{1:K},x_{1:K})dx_{1:K}\\
&=\int f(y_{1:K}|x_{1:K})f(x_{1:K})dx_{1:K},
\end{align*}
where $f(x_{1:K})$ denotes the probability distribution of $x_{1:K}$.  In~particular:
\begin{align*}
y_{1:K}|x_{1:K} &\sim \mathcal{N}(c(x_{1:K}), S_{1:K}\Sigma(v_{1:K})S_{1:K}^\top)\\
x_{1:K} &\sim \mathcal{N}(\mu_{1:K}, \Sigma(x_{1:K})).
\end{align*}
Therefore, the best Gaussian approximation to the marginal distribution of $y_{1:K}$ is:
\begin{align*}
y_{1:K} &\sim \mathcal{N}(\mathbb{E}(y_{1:K}), \Sigma(y_{1:K}))\\
\mathbb{E}(y_{1:K})&=c(\mu_{1:K})\\
\Sigma(y_{1:K})&= S_{1:K}\Sigma(v_{1:K})S_{1:K}^\top+C(\tilde{\mu}_{1:K})\Sigma(x_{1:K})C(\tilde{\mu}_{1:K})^\top.
\end{align*}
Thus, from Lemma~\ref{lem:entropy_gaussian},~\eqref{aux_4_app} follows.
\end{proof}

\begin{myremark}
The time complexity for the evaluation of $\mathbb{H}(x_{1:K}|y_{1:K})$ using Lemma~\ref{lem:entropy_closed_Sigma} is decided by the sparsity of $\Sigma(x_{1:K})$ since the rest of the matrices are block diagonal. 
\end{myremark}

We complete the proof for each case of Theorem \ref{th:time_compl_evaluation}:

\begin{itemize}
\item \emph{Time complexity of each evaluation of $\mathbb{H}(x_{1:K}|y_{1:K})$ when either $\Sigma(x_{1:K})$ or $\Sigma(x_{1:K})^{-1}$ is exactly sparse (that is, block tri-diagonal)}:  We present the proof only for the case where $\Sigma(x_{1:K})^{-1}$ is exactly sparse since the proof for the case where $\Sigma(x_{1:K})$ is exactly sparse is similar.  In particular, consider the formula of $\mathbb{H}(x_{1:K}|y_{1:K})$ in Lemma~\ref{lem:entropy_closed_inv_sigma}: $T_1'$ involves the log determinant of a matrix that is the sum of two $nK\times nK$ sparse matrices: the first matrix is block diagonal, and the second one is block tri-diagonal.  The block diagonal matrix is evaluated in $O(n^{2.4}K)$ time, since the determinant of an $n \times n$ matrix is computed in $O(n^{2.4})$ time using the Coppersmith-Winograd algorithm~\cite{coppersmith1987matrix}.  Then, $T_1'$ is evaluated in $O(n^{2.4}K)$~\cite[Theorem~2]{molinari2008determinants}.

\item \emph{Time complexity of each evaluation of $\mathbb{H}(x_{1:K}|y_{1:K})$ when both $\Sigma(x_{1:K})$ and $\Sigma(x_{1:K})^{-1}$ are dense}:  In this case, $T_1'$ (and similarly $T_2$ in Lemma~\ref{lem:entropy_closed_Sigma}) is the log determinant of a dense $nK \times nK$ matrix.  Therefore, it is evaluated in $O((nK)^{2.4})$ time, since the determinant of an $n \times n$ matrix is computed in $O(n^{2.4})$ time using the Coppersmith-Winograd algorithm~\cite{coppersmith1987matrix}.\qedhere
\end{itemize}
\end{proof}

\bibliographystyle{IEEEtran}
\bibliography{references}

\end{document}